\documentclass[aps, prb, reprint, superscriptaddress, amsmath, amssymb]{revtex4-2}

\usepackage{braket}
\usepackage{graphicx}
\usepackage{dcolumn}
\usepackage{bm}
\usepackage{color}
\usepackage{soul} 
\usepackage{mathtools}
\usepackage[colorlinks, citecolor=blue,urlcolor=blue,linkcolor=blue]{hyperref}
\usepackage{float}
\usepackage{amsthm}
\theoremstyle{plain} 
\newtheorem{theorem}{Theorem}[section]     
\newtheorem{lemma}[theorem]{Lemma}         
\newtheorem{proposition}[theorem]{Proposition}

\theoremstyle{remark}

\begin{document}

\title{Global bifurcations and basin geometry of the nonlinear non-Hermitian skin effect}

\author{Heng Lin}
    \affiliation{State Key Laboratory of Low-Dimensional Quantum Physics and Department of Physics, Tsinghua University, Beijing 100084, China}
\author{Yunyao Qi}
    \affiliation{State Key Laboratory of Low-Dimensional Quantum Physics and Department of Physics, Tsinghua University, Beijing 100084, China}
\author{Gui-Lu Long}%
    \email{gllong@tsinghua.edu.cn}
    \affiliation{State Key Laboratory of Low-Dimensional Quantum Physics and Department of Physics, Tsinghua University, Beijing 100084, China}
    \affiliation{Frontier Science Center for Quantum Information, Beijing 100084, China}
    \affiliation{Beijing Academy of Quantum Information Sciences, Beijing 100193, China}

\begin{abstract}
We study a continuum Hatano--Nelson model with a saturating nonlinear nonreciprocity and analyze its stationary states via the associated phase-space flow.
We uncover a global scenario controlled by a subcritical Hopf bifurcation and a saddle-node of limit cycles, which together generate a finite coexistence window.
In this window, skin modes and extended states are both stable at a fixed energy $E$, separated by a nonlinear basin separatrix in phase space rather than a spectral (mobility-edge) mechanism in a linear system.
An averaged amplitude equation yields closed-form predictions for the limit-cycle branches and the SNLC threshold.
Building on the basin geometry, we introduce a basin-fraction order parameter that exhibits a first-order-like jump at SNLC. 
Intriguing physical phenomena in the coexistence window are also revealed, such as separatrix-induced long-lived spatial transients and hysteresis.
Overall, our findings highlight that, beyond linear spectral concepts, global attractor--basin geometry provides a powerful and complementary lens for understanding stationary states in nonlinear non-Hermitian systems.
\end{abstract}

\maketitle

\section{Introduction}\label{sec:intro}

Non-Hermitian physics provides a unified language for open and nonconservative systems, where gain, loss, dissipation, or asymmetric couplings render the effective Hamiltonian non-Hermitian~\cite{El2018Non}. 
The past decade has witnessed an explosion of interest in this field, leading to the discovery of singular behaviors and novel potential applications without Hermitian counterparts~\cite{PhysRevLett.77.570, PhysRevB.56.8651, PhysRevB.58.8384, PhysRevLett.80.5243, berry2004physics, PhysRevLett.102.065703, PhysRevLett.103.093902, Ruter2010Observation, Feng2011Nonreciprocal,  Regensburger2012Parity-time, Zhen2015Spawning,  Poli2015Selective, longhi2015robust, Chen2017Exceptional, Xiao2017Observation, Bliokh2019Topological, Özdemir2019Parity–time, ashida2020non, Leefmans2022Topological, Arkhipov2023Dynamically, PhysRevResearch.5.L032026}.
A particularly striking phenomenon is the non-Hermitian skin effect (NHSE). In the NHSE, an extensive number of bulk eigenstates become sensitive to boundaries and accumulate at the edges under open boundary conditions, in sharp contrast to the extended Bloch waves of Hermitian systems~\cite{PhysRevB.97.121401, prl_wang, PhysRevLett.121.136802}. 
This boundary accumulation also underlies the breakdown of conventional bulk--boundary correspondence and has motivated the establishment of non-Bloch band theory~\cite{prl_wang, PhysRevLett.123.066404,PhysRevLett.125.126402, PhysRevX.14.021011, okuma2023non, PhysRevLett.124.086801, PhysRevX.8.031079}. 
These developments, together with growing experimental realizations across optical~\cite{doi:10.1126/science.aaz8727, xiao2020non, Liu2022Complex, Leefmans2024Topological}, mechanical~\cite{brandenbourger2019non,palacios2021guided,PhysRevLett.131.207201}, circuit~\cite{helbig2020generalized,zhang2021observation,10.21468/SciPostPhys.16.1.002} and cold atom~\cite{PhysRevLett.124.070402} platforms, have established the NHSE as a central organizing principle for non-Hermitian phenomena.

Alongside these advances, the interplay between non-Hermiticity and nonlinearity has emerged as a rapidly growing frontier~\cite{Bahari2017Nonreciprocal, harari2018topological, bandres2018topological, smirnova2020nonlinear, cai2025versatile, benzaouia2022nonlinear, wang2019dynamics, bai2023nonlinearity, PhysRevB.109.085311, PhysRevResearch.6.013148, PhysRevLett.134.133801, kwong2026, yuce2021nonlinear, PhysRevB.109.094308, PhysRevLett.134.243805, hopf_bif,  hamanaka2026, PhysRevLett.100.030402, RevModPhys.88.035002, PhysRevLett.123.253903, bocharov2023topological, hashemi2025reconfigurable, cheng2025solitons, Pontula2025Non, Salcedo2025Demonstration, Reisenbauer2024Non}. 
This broader direction spans, for example, nonlinear topological photonics and lasers~\cite{Bahari2017Nonreciprocal, harari2018topological, bandres2018topological, smirnova2020nonlinear, cai2025versatile}, 
nonlinear exceptional-point physics~\cite{benzaouia2022nonlinear, wang2019dynamics, bai2023nonlinearity, PhysRevB.109.085311, PhysRevResearch.6.013148, PhysRevLett.134.133801, kwong2026}, 
as well as nonlinear extensions of NHSE~\cite{yuce2021nonlinear, PhysRevB.109.094308, PhysRevLett.134.243805, hopf_bif, hamanaka2026} 
and non-Hermitian solitons~\cite{PhysRevLett.100.030402, RevModPhys.88.035002, PhysRevLett.123.253903, bocharov2023topological, hashemi2025reconfigurable, cheng2025solitons}.

Among these directions, a recent study revealed a Hopf-bifurcation scenario for the nonlinear skin effect by analyzing a nonlinear nonreciprocal Hatano--Nelson model in a spatial-dynamics formulation, where the spatial coordinate is treated as an evolution parameter in a two-dimensional phase-space flow~\cite{hopf_bif}.
Within this framework, they identified a \emph{supercritical} Hopf bifurcation at which a fixed point loses stability and a stable limit cycle emerges, yielding a nonlinearity-driven skin-to-extended transition~\cite{hopf_bif}.
These results naturally raise broader questions beyond this local Hopf picture:
\emph{(i)} what qualitatively changes if the amplitude-dependent nonreciprocity is \emph{non-monotonic} and saturating, so that the effective gain is suppressed at large amplitudes;
\emph{(ii)} what other dynamical bifurcations beyond a supercritical Hopf (e.g., subcritical Hopf or saddle node of limit cycles~\cite{kuznetsov1998elements, guckenheimer2013nonlinear, thompson2002nonlinear}) can occur, and what stationary localization phenomena do they imply;
and \emph{(iii)} can such nonlinear non-Hermitian models host richer global phase-space structures, in particular the coexistence of multiple stable attractors under identical parameters, leading to history- and preparation-dependent stationary states?

In this work, we address the questions above in a continuum nonlinear Hatano--Nelson model~\cite{PhysRevB.56.8651, PhysRevLett.77.570}  with a \emph{saturating} amplitude-dependent nonreciprocity.
We uncover a bifurcation scenario that goes beyond a single supercritical Hopf:
the origin undergoes a \emph{subcritical} Hopf bifurcation at $\gamma=0$, while a \emph{saddle-node of limit cycles (SNLC)} occurs at $\gamma=\gamma_c<0$.
Together these two bifurcations generate a finite coexistence window $\gamma_c<\gamma<0$, where the bistable attractors cause the coexistence of skin and extended states.
Unlike localized--extended coexistence in linear systems, which is typically tied to spectral structures such as mobility edges~\cite{evers2008anderson, Kramer_1993, PhysRevLett.104.070601, PhysRevLett.114.146601, qi2026mobilityedge}, here the two stationary states coexist at a fixed $E$ and are separated by a nonlinear basin separatrix.
Then we derive an averaged amplitude equation, which not only provides an approximate yet simple analytical explanation for the bifurcation diagram, but also yields closed-form predictions for the limit-cycle branches and the SNLC threshold $\gamma_c$.
Building on the separatrix geometry, we further introduce a basin-fraction order parameter, which turns bistability into a phase diagram and reveals a first-order-like jump at SNLC. 
Finally, intriguing physical phenomena related to the phase-space basin geometry in the coexistence window are also revealed, such as separatrix-induced long-lived spatial transients and hysteresis.

This paper is organized as follows.
Sec.~\ref{sec:model} introduces the saturating nonlinear Hatano--Nelson model and formulates the stationary problem as a planar phase-space flow.
Sec.~\ref{sec:bif_diag} presents the global bifurcation diagram and related skin and extended stationary states.
In Sec.~\ref{sec:ave_eqn}, we derive an averaged amplitude equation that provides analytical predictions for the limit-cycle branches and the SNLC threshold.
Sec.~\ref{sec:basin_frac} introduces a basin-fraction order parameter and constructs a phase diagram.
In Sec.~\ref{sec:coexist_phys}, we discuss physical consequences of bistability, including long-lived transients and hysteresis.
Finally, we conclude in Sec.~\ref{sec:concl}.

\section{Model}\label{sec:model}

We study a nonlinear non-Hermitian extension of the continuum Hatano--Nelson~\cite{PhysRevB.56.8651, PhysRevLett.77.570} model, in which the nonreciprocal (non-Hermitian) term is modulated by the local wave amplitude. Concretely, we consider a stationary nonlinear Schr\"odinger equation
\begin{equation}
    \hat H(\psi)\,\psi = E\,\psi ,
    \label{eq:nl_evp}
\end{equation}
with the nonlinear Hamiltonian
\begin{equation}
    \hat H(\psi)=\frac{\hat k^2}{2}+ i\,F(|\psi|^2)\,\hat k,\qquad 
    \hat k \coloneqq -i\partial_x ,
    \label{eq:H_def}
\end{equation}
where the amplitude-dependent non-Hermitian term is chosen in a cubic-quintic form,
\begin{equation}
    F(z)=\gamma + a z - b z^2,\qquad z\equiv |\psi|^2 .
    \label{eq:F_def}
\end{equation}
Here $\gamma\in\mathbb{R}$ quantifies the linear nonreciprocity, while $a>0$ and $b>0$ encode a saturating nonlinear modulation: the $+a|\psi|^2$ term enhances the effective nonreciprocity at moderate amplitudes, whereas the $-b|\psi|^4$ term suppresses it at large amplitudes. For reference, in the linear limit $a=b=0$, the model reduces to the conventional continuum nonreciprocal Hatano--Nelson model: for $\gamma<0$ the wave function decays exponentially and localizes near the left boundary, whereas for $\gamma>0$ it grows exponentially and localizes near the right boundary~\cite{okuma2023non, PhysRevLett.124.086801, PhysRevX.8.031079}.
Throughout the paper, we focus on positive energy $E>0$ and real stationary states $\psi(x)\in\mathbb{R}$ for clarity and to facilitate the dynamical-systems analysis, so that $|\psi|^2=\psi^2$.
This setting already captures the bifurcation structure central to our results. 
Extending to complex wave functions leads to a higher-dimensional flow and is left for future work.

\subsection{Stationary equation as a phase-space flow}

Eq.~\eqref{eq:nl_evp} can be written as an ordinary differential equation in real space:
\begin{equation}
    \partial_x^2\psi -2F(\psi^2)\,\partial_x\psi + 2E\,\psi = 0.
    \label{eq:stationary_ode}
\end{equation}
Following the method used in~\cite{hopf_bif}, a key viewpoint of this work is to treat $x$ as an evolution parameter and characterize stationary states by the long-$x$ behavior of the associated phase-space flow.
Introducing
\begin{equation}
    v(x) \coloneqq \partial_x\psi(x),
    \label{eq:v_def}
\end{equation}
Eq.~\eqref{eq:stationary_ode} becomes the planar dynamical system
\begin{equation}
\begin{aligned}
    \partial_x \psi &= v,\\
    \partial_x v &= 2F(\psi^2)\,v - 2E\,\psi.
\end{aligned}
\label{eq:phase_flow}
\end{equation}
In this formulation, a stationary state corresponds to a trajectory $(\psi(x),v(x))$ in phase space.

To investigate the long-$x$ behavior of the wave function, we consider a semi-infinite boundary condition in this paper, with $\psi(0) = 0$ and parametrize the remaining freedom by the boundary slope
\begin{equation}
    s \coloneqq \partial_x\psi(0) = v(0) \in\mathbb{R}.
    \label{eq:bc_slope}
\end{equation}
Thus, for each choice of parameters $(\gamma,a,b,E)$, the family of trajectories launched from $(\psi,v)=(0,s)$ provides a natural ensemble of stationary states.
Operationally, we use a standard shooting method, integrating Eq.~\eqref{eq:phase_flow} forward in $x$.
Details of the numerical methods used in this paper are provided in Appendix~\ref{app:numerics}.

\subsection{Attractor-localization correspondence}
The phase-space flow picture immediately establishes a correspondence between the attractor types in phase space and the localization properties of the wave function $\psi(x)$.

First, the origin $(\psi,v)=(0,0)$ is the unique fixed point of Eq.~\eqref{eq:phase_flow}.
Linearizing Eq.~\eqref{eq:phase_flow} around the origin~\cite{kuznetsov1998elements, guckenheimer2013nonlinear} gives
\begin{equation}
    \partial_x 
    \begin{pmatrix}
    \psi \\ v
    \end{pmatrix}
    =
    \underbrace{
    \begin{pmatrix}
    0 & 1 \\
    -2E & 2\gamma
    \end{pmatrix}}_{\coloneqq J(0)}
    \begin{pmatrix}
    \psi \\ v
    \end{pmatrix},
    \label{eq:lin_system}
\end{equation}
whose characteristic equation $\det\left(\lambda \mathbf{I}-J(0)\right)=0$ yields the eigenvalues
\begin{equation}
    \lambda = \gamma \pm \sqrt{\gamma^2-2E}.
    \label{eq:lin_origin}
\end{equation}
For $E>0$, $\mathrm{Re}(\lambda)$ has the same sign as $\gamma$, so the origin is locally attracting (stable) for $\gamma<0$ and repelling (unstable) for $\gamma>0$.
A trajectory attracted to the origin corresponds to a state whose amplitude decays as $x$ increases, i.e., a \emph{skin-localized} state in real space.

Beyond the fixed point, the flow may admit limit cycles, i.e., closed periodic orbits in phase space, which can be attracting or repelling~\cite{kuznetsov1998elements, guckenheimer2013nonlinear}.
A stable limit cycle $\Gamma$ corresponds to a bounded oscillatory state with a nonvanishing asymptotic amplitude as $x\to\infty$, i.e., an \emph{extended} state in real space.

When multiple attractors (e.g., the origin and a stable limit cycle) are present, their basins of attraction are separated by a separatrix, which in our setting is typically organized by an unstable limit cycle~\cite{kuznetsov1998elements, guckenheimer2013nonlinear}.
As a result, initial slopes $s$ on different sides of this separatrix evolve toward different long-$x$ behaviors (skin vs.\ extended) even at identical system parameters.
This attractor-based viewpoint will serve as the backbone of the bifurcation diagram in Sec.~\ref{sec:bif_diag} and the basin-measure phase diagram in Sec.~\ref{sec:basin_frac}.

\section{Global bifurcation diagram}\label{sec:bif_diag}

The phase-space flow formulation turns the stationary state problem into a planar dynamical system.
The long-$x$ stationary behaviors are organized by the invariant sets (fixed points and limit cycles) and their stability.
Figure~\ref{fig:bif_diag} summarizes this organization of our system in a global bifurcation diagram~\cite{kuznetsov1998elements} as $\gamma$ is varied, using the limit-cycle amplitude
$A$ defined as the turning-point value $|\psi|$ at $v=\partial_x\psi=0$ on the cycle. The branch $A=0$ corresponds to the origin.

Two bifurcations control the entire scenario.
First, the origin is stable for $\gamma<0$ and loses stability at $\gamma=0$ (Sec.~\ref{sec:model}).
In our model, this occurs via a \emph{subcritical Hopf} bifurcation~\cite{kuznetsov1998elements, guckenheimer2013nonlinear}: 
an \emph{unstable} small-amplitude limit cycle exists on the $\gamma<0$ side and shrinks into the origin as $\gamma\to 0^{-}$.
Second, at $\gamma=\gamma_c<0$ a \emph{SNLC bifurcation} creates a pair of finite-amplitude periodic orbits: an \emph{outer stable} cycle and an \emph{inner unstable} cycle~\cite{kuznetsov1998elements}.
These two bifurcations partition the $\gamma$ axis into three dynamical regimes, as shown in Fig.~\ref{fig:bif_diag}.
For $\gamma<\gamma_c$, the origin is the only attractor (skin-only).
For $\gamma>0$, the origin is unstable and the stable limit cycle is the unique attractor (extended-only).
Most importantly, in the intermediate window $\gamma_c<\gamma<0$ (shaded in Fig.~\ref{fig:bif_diag}), the stable origin and the stable outer limit cycle \emph{coexist}, forming a bistable regime.
A proof sketch of this bifurcation structure is provided in Sec.~\ref{sec:bif_diag_C}.

\begin{figure}[t]
  \includegraphics[width=0.95\linewidth]{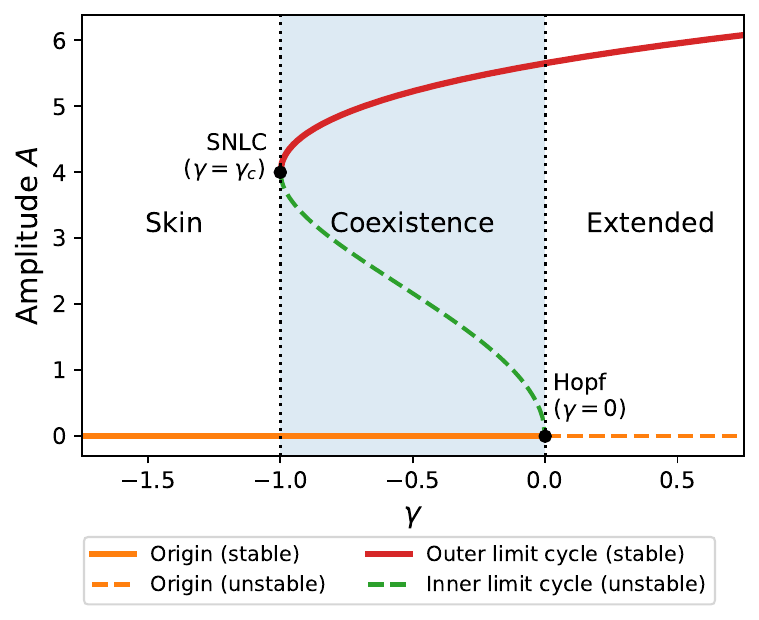}
  \caption{Global bifurcation diagram of the phase-space flow Eq.~\eqref{eq:phase_flow}, plotted in terms of the limit-cycle amplitude $A$ versus $\gamma$. Solid (dashed) curves denote stable (unstable) invariant sets. The shaded interval $\gamma_c<\gamma<0$ indicates bistability, where the fixed point (skin) and the limit cycle (extended) coexist. Parameters are $a=1/2$, $b=1/32$, and $E=8$.
  \label{fig:bif_diag} 
  }
\end{figure}

The bifurcation diagram above identifies which invariant sets exist and are stable as $\gamma$ is varied, establishing the dynamical backbone of this system.
To connect this backbone to the non-Hermitian skin physics, we next translate these invariant sets into stationary wave functions, highlighting their localization properties and the intuitive mechanisms underlying the three regimes.

Figure~\ref{fig:phase} shows representative phase-space portraits together with wave-function profiles.
In the \textbf{skin regime} ($\gamma<\gamma_c$), trajectories spiral toward the origin [Fig.~\ref{fig:phase}(a)], and the wave function decays with $x$ [Fig.~\ref{fig:phase}(b)].
This behavior is consistent with the linear Hatano--Nelson model: the negative nonreciprocal term causes decay away from the left boundary, leading to left-boundary skin localization~\cite{okuma2023non, PhysRevLett.124.086801, PhysRevX.8.031079}.
In our nonlinear model, sufficiently negative $\gamma$ keeps the nonreciprocity $F(\psi^2)=\gamma+a\psi^2-b\psi^4$ negative for any $\psi$, leading to this decay and making the origin the sole attractor.

In the \textbf{extended regime} ($\gamma>0$), the long-$x$ behavior is governed by a stable limit cycle [Fig.~\ref{fig:phase}(h)], and the wave function approaches a bounded oscillatory state with a nonzero amplitude [Fig.~\ref{fig:phase}(i)].
A useful intuition is ``growth with saturation'': at small amplitudes the effective nonreciprocity $F(\psi^2)=\gamma+a\psi^2-b\psi^4$ is positive and promotes amplification, while at larger amplitudes the $-b\psi^4$ term suppresses $F$ and prevents runaway growth, pinning the motion to a finite-amplitude cycle.

The most distinctive behavior occurs in the \textbf{coexistence regime} ($\gamma_c<\gamma<0$), where both the origin and an outer limit cycle are attracting, with an unstable inner limit cycle serving as a basin separatrix [Fig.~\ref{fig:phase}(c)].
Consequently, two trajectories launched from different boundary slopes can evolve to completely different long-$x$ stationary states.
A trajectory inside the inner cycle will be attracted to the origin, forming a skin mode [Fig.~\ref{fig:phase}(d)], while a trajectory outside the inner cycle is attracted to the stable outer cycle, corresponding to an extended state [Fig.~\ref{fig:phase}(g)].
In linear systems, such localized--extended coexistence often depends on spectral structures (such as mobility edges)~\cite{evers2008anderson, Kramer_1993, PhysRevLett.104.070601, PhysRevLett.114.146601, qi2026mobilityedge}; however, in our system, the skin and extended states exist simultaneously at a fixed $E$, and the separatrix is provided by the unstable inner cycle.

\begin{figure*}[t]
  \includegraphics[width=1.0\linewidth]{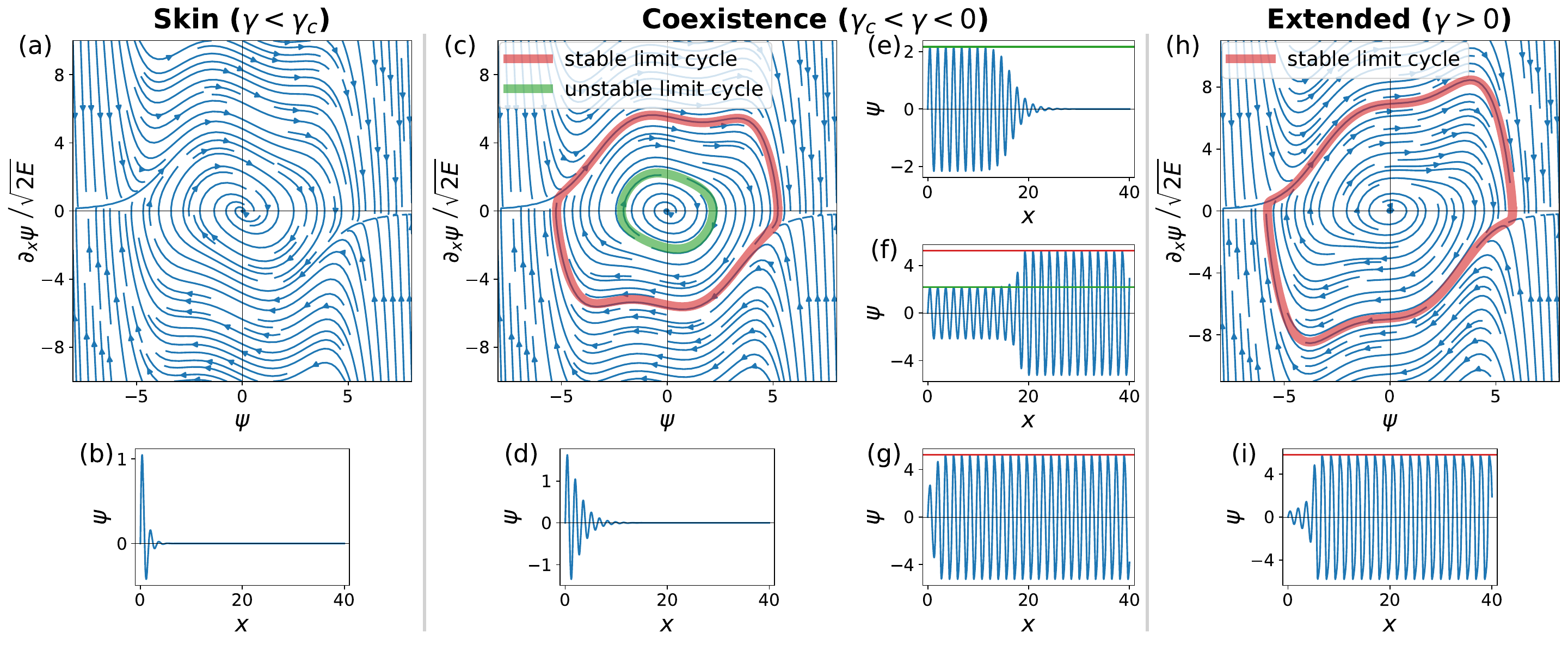}
  \caption{\textbf{Representative phase-space flow portraits and wave function profiles across the three regimes in Fig.~\ref{fig:bif_diag}.}
  Phase-space flows are shown in $(\psi,\partial_x\psi)$ with the vertical axis normalized by the natural frequency $\sqrt{2E}$.
  \textbf{Skin regime} $\;(\gamma<\gamma_c)$: (a) flow at $\gamma=-1.2$; (b) a typical skin mode with $\partial_x \psi(0)=6$.
  \textbf{Coexistence regime} $\;(\gamma_c<\gamma<0)$: (c) flow at $\gamma=-0.5$, showing a stable outer limit cycle (red) and an unstable inner limit cycle (green); (d-g) representative profiles with $\partial_x \psi(0)=7$ (skin), $\partial_x \psi(0)=8.69755$ (near-separatrix trajectory returning to the skin attractor), $\partial_x \psi(0)=8.69756$ (near-separatrix trajectory captured by the extended attractor), and $\partial_x \psi(0)=10$ (extended), respectively.
  \textbf{Extended regime} $\;(\gamma>0)$: (h) flow at $\gamma=0.2$, featuring a stable limit cycle (red); (i) a representative extended state with $\partial_x\psi(0)=2$.
  For (e-g) and (i), guide lines (green/red) indicate the inner/outer cycle amplitudes predicted in Sec.~\ref{sec:ave_eqn}.
  For all subfigures, parameters are $a=1/2$, $b=1/32$, $E=8$.
  \label{fig:phase} 
  }
\end{figure*}

\subsection{Proof sketch: why the diagram has this structure}\label{sec:bif_diag_C}

Here we give a concise proof sketch supporting the bifurcation diagram in Fig.~\ref{fig:bif_diag}. Full technical details are provided in Appendix~\ref{app:proofs}.

\emph{(i) Trapping and boundedness.}
The saturating nonlinearity $F(\psi^2)=\gamma+a\psi^2-b\psi^4$ implies that for sufficiently large $|\psi|$ one has $F(\psi^2)<0$ (the $-b\psi^4$ term dominates).
This provides strong effective damping in the $v$-equation of Eq.~\eqref{eq:phase_flow}, allowing one to construct a compact trapping region in the $(\psi,v)$ plane: all trajectories enter it after finite $x$ and remain bounded thereafter.

\emph{(ii) Strong negative $\gamma$.}
When $\gamma$ is sufficiently negative, the nonreciprocal term $F(\psi^2)$ stays negative for any $\psi$. In this case, the flow is everywhere effectively damped, and one can show the periodic orbits are absent and the origin is globally attracting.

\emph{(iii) Positive $\gamma$.} 
In contrast, for $\gamma>0$ the origin is unstable (Sec.~\ref{sec:model}), while boundedness excludes the possibility of divergence; in a planar flow this forces the long-$x$ behavior toward a periodic attractor. To be precise, the Li\'enard theorem~\cite{lienard1928, guckenheimer2013nonlinear} provides \emph{existence and uniqueness} of a stable limit cycle for $\gamma>0$.

\emph{(iv) Subcritical Hopf at $\gamma=0$.}
At $\gamma=0$, $\mathrm{tr}\left( J(0) \right)=2\gamma$ changes sign while $\det \left( J(0) \right)=2E>0$, so the equilibrium undergoes a Hopf bifurcation with frequency $\sqrt{2E}$.
For $a>0$, the first Lyapunov coefficient indicates that the Hopf bifurcation is \emph{subcritical}~\cite{kuznetsov1998elements}. This implies that an \emph{unstable} small-amplitude limit cycle exists for $\gamma\to 0^{-}$ and collides with the origin at $\gamma=0$.

\emph{(v) Why a SNLC at $\gamma_c<0$ and bistability for $\gamma_c<\gamma<0$.}
From (iii) and (iv), a stable cycle exists for $\gamma>0$ and an unstable cycle exists just below $\gamma=0$, whereas from (ii) no cycle can exist for sufficiently negative $\gamma$.
Thus, upon continuation from $0$ to more negative $\gamma$, the limit cycle branches must terminate at some $\gamma=\gamma_c<0$.
Generically, a one-parameter family of planar limit cycles terminates through a SNLC bifurcation: a pair of cycles with opposite stability collides and annihilates each other.
This yields the SNLC at $\gamma=\gamma_c$ in Fig.~\ref{fig:bif_diag} with an inner unstable cycle and an outer stable cycle.
Since the origin remains stable for $\gamma<0$, the interval $\gamma_c<\gamma<0$ is necessarily bistable, with basins separated by the unstable cycle.

\section{Averaged amplitude equation}\label{sec:ave_eqn}

In Sec.~\ref{sec:bif_diag}, we established the global bifurcation skeleton of the phase-space dynamics. 
In this section, we develop a complementary analytical description based on averaging method~\cite{2007Averaging, Kevorkian1996Multiple}, which yields a simple \emph{one-dimensional} amplitude equation capturing the emergence, stability, and amplitudes of the limit cycles. 
Moreover, this approach predicts the inner/outer cycle amplitudes and the SNLC threshold $\gamma_c$, in quantitative agreement with numerical continuation (Fig.~\ref{fig:avg_eqn}).

\subsection{Derivation of the averaged amplitude equation}

In the Hermitian limit, our model reduces to a harmonic oscillator with frequency $\omega=\sqrt{2E}$. It is therefore natural to scale the variable $v=\partial_x \psi$ in phase space [Eq.~\eqref{eq:phase_flow}], 
\begin{equation}
w \coloneqq v / \omega.
\end{equation}
We then parameterize the trajectory $(\psi, w)$ by polar coordinates $(r, \theta)$ in phase space,
\begin{equation}
\psi=r\cos\theta,\qquad w=-r\sin\theta,
\label{eq:polar_def}
\end{equation}
so that $r\ge 0$ measures the oscillation amplitude while $\theta$ is the oscillation phase.
A substitution of Eq.~\eqref{eq:polar_def} into Eq.~\eqref{eq:phase_flow} yields an \emph{exact} amplitude equation of the form
\begin{equation}
\partial_x r = 2r\Bigl(\gamma+a r^2\cos^2\theta-b r^4\cos^4\theta\Bigr)\sin^2\theta,
\label{eq:r_exact}
\end{equation}
and a phase equation
\begin{equation}
\partial_x \theta =\omega + \mathcal{O}\!\left(|\gamma|+a r^2+b r^4\right),
\label{eq:theta_exact}
\end{equation}
where the $\mathcal{O}(\cdot)$ correction comes from the nonlinear/non-Hermitian drift term in Eq.~\eqref{eq:phase_flow}.
Eq.~\eqref{eq:r_exact} shows that the amplitude change is modulated by trigonometric factors, i.e., it oscillates within each cycle.

The key approximation is the standard ``slowly varying amplitude'' condition~\cite{2007Averaging, Kevorkian1996Multiple}:
\begin{equation}
|\gamma|+a r^2+b r^4 \ll \omega=\sqrt{2E},
\label{eq:fast_slow_cond}
\end{equation}
equivalently, the non-Hermitian modulation
$F(\psi^2)=\gamma+a\psi^2-b\psi^4$
is small compared with the natural oscillation frequency scale.
Under Eq.~\eqref{eq:fast_slow_cond}, the phase $\theta$ advances almost uniformly over one period, while $r$ changes only slightly. 
Thus, to leading order one may average Eq.~\eqref{eq:r_exact} over $\theta\in[0,2\pi)$, replacing the rapidly oscillating factors by their mean values, which leads to the \emph{averaged amplitude equation}
\begin{equation}
\partial_x r = h(r),\qquad 
h(r)\equiv r\left(\gamma+\frac{a}{4}r^2-\frac{b}{8}r^4\right).
\label{eq:amplitude_avg}
\end{equation}
It is a one-dimensional equation of $r\ge 0$ and will be the main tool of this section.
A detailed derivation of Eq.~\eqref{eq:amplitude_avg} is provided in Appendix~\ref{app:avg_eqn}.

\subsection{Invariant sets, stability, and bifurcations from the averaged amplitude equation}
\label{sec:ave_interpret}

Eq.~\eqref{eq:amplitude_avg} provides a compact reduced description of the oscillation amplitude of the wave function [Eq.~\eqref{eq:polar_def}].
Crucially, since Eq.~\eqref{eq:amplitude_avg} is one-dimensional, the existence, stability, and disappearance of attractors reduce to a simple algebraic problem, which not only explains the bifurcation diagram in Sec.~\ref{sec:bif_diag}, but also yields theoretical predictions for the SNLC threshold $\gamma_c$ and the limit-cycle amplitudes.

In particular, equilibria $r=A$ of Eq.~\eqref{eq:amplitude_avg} (i.e., $h(A)=0$) correspond to invariant sets of the phase-space flow.
The stability of an invariant set is also determined by the sign of the derivative $h'(r)$: if $h'(A) < 0$, the invariant set is stable, and vice versa~\cite{guckenheimer2013nonlinear}.

We first consider the trivial equilibrium $A_0=0$, which corresponds to the fixed point at the origin $(\psi,v)=(0,0)$. 
From Eq.~\eqref{eq:amplitude_avg}, $h'(0)=\gamma$, so the origin is stable for $\gamma<0$ and unstable for $\gamma>0$, consistent with the linearized analysis in Sec.~\ref{sec:model}.
Moreover, near $\gamma=0$ the expansion
\begin{equation}
h(r)=\gamma r+\frac{a}{4}r^{3}+\mathcal{O}(r^{5}).
\end{equation}
shows that for $a>0$ the cubic term is positive, implying that the small-amplitude cycle on the $\gamma<0$ side is repelling; this is the normal-form signature of a \emph{subcritical} Hopf bifurcation at $\gamma=0$~\cite{kuznetsov1998elements}.

We next turn to nontrivial equilibria $A>0$, which determine the inner/outer limit-cycle branches and the SNLC threshold.
They satisfy
\begin{equation}
\gamma+\frac{a}{4}A^2-\frac{b}{8}A^4=0,
\label{eq:nontrivial_eq}
\end{equation}
whose real positive solutions are
\begin{equation}\label{eq:amp_predic}
\begin{aligned}
&A_{\rm in}(\gamma)=\sqrt{\frac{a}{b}\left(1-\sqrt{1+\frac{8b\gamma}{a^2}}\right)},\\
&A_{\rm out}(\gamma)=\sqrt{\frac{a}{b}\left(1+\sqrt{1+\frac{8b\gamma}{a^2}}\right)},
\end{aligned}
\end{equation}
corresponding to the amplitudes of inner and outer limit cycles.
For $A>0$ with $h(A)=0$, differentiating Eq.~\eqref{eq:amplitude_avg} yields
\begin{equation}
h'(A)=\frac{b}{2}A^{2}\left(\frac{a}{b}-A^{2}\right).
\label{eq:hprime_on_branch}
\end{equation}
Therefore, $h'(A_{\rm out})<0$ implies a stable outer cycle, whereas $h'(A_{\rm in})>0$ implies an unstable inner cycle.

These two real positive solutions exist if and only if $1+8b\gamma/a^{2}\ge 0$, i.e., $\gamma \ge \gamma_c^{\rm (th)}$, where $\gamma_c^{\rm (th)}$ gives the first-order prediction for the SNLC threshold $\gamma_c$:
\begin{equation}
\gamma_c^{\rm (th)} = -\frac{a^2}{8b}.
\label{eq:gamma_c_th}
\end{equation}
At $\gamma=\gamma_c^{\rm (th)}$, the two branches merge at a degenerate positive root, which corresponds to a SNLC bifurcation with finite-amplitude nucleation.

\begin{figure}[H]
  \includegraphics[width=0.95\linewidth]{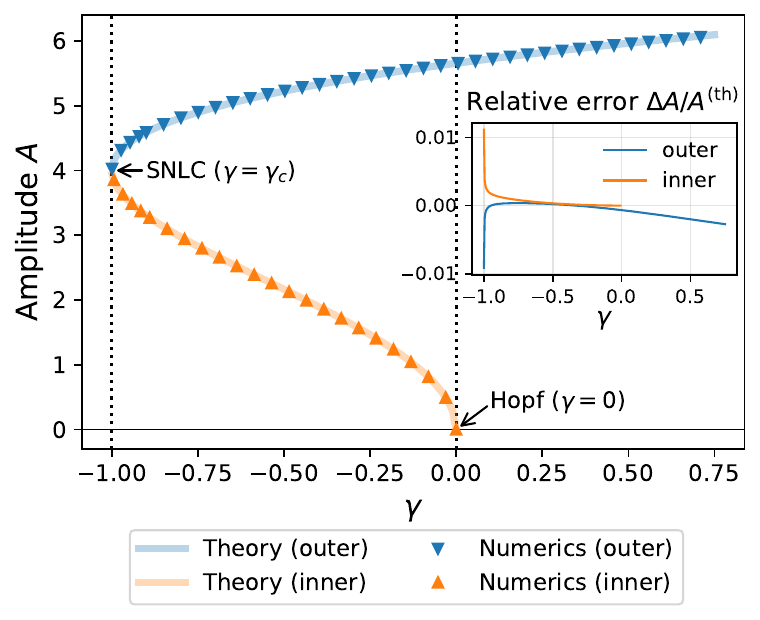}
  \caption{Comparison between the limit-cycle amplitudes predicted by the averaging theory and those obtained from numerical continuation of the phase-space flow dynamics. Inset: relative deviation $(A^{\rm (num)}-A^{\rm (th)})/A^{\rm (th)}$ for the outer and inner branches. Parameters are $a=1/2$, $b=1/32$, $E=8$, and the theoretical prediction of $\gamma_c$ is $\gamma_c^{\rm (th)} = -a^2/8b = -1$.
  \label{fig:avg_eqn} 
  }
\end{figure}

Finally, Fig.~\ref{fig:avg_eqn} compares the theoretical amplitudes $A^{\rm (th)}$ predicted by Eq.~\eqref{eq:amp_predic} with the numerical continuation results $A^{\rm (num)}$.
We observe excellent agreement across the entire parameter range, and the theoretical SNLC threshold $\gamma_c^{\rm (th)}$ also matches the numerical result, confirming that the averaged equation captures the bifurcation structure.

\section{Basin-fraction order parameter and phase diagram}\label{sec:basin_frac}

In a nonlinear system, the notion of a linear spectrum is generally inapplicable, and stationary states are more naturally classified by the long-$x$ attractors of the phase-space flow dynamics Eq.~\eqref{eq:phase_flow}. 
Section~\ref{sec:bif_diag} established the global bifurcation skeleton in terms of invariant sets (the fixed point and limit cycles). 
Here we introduce a complementary \emph{basin-geometry} characterization that turns the bistability window into a quantitative phase diagram~\cite{menck2013basinstability, schultz2017basinstability, vankan2016constrainedbasin}.

Under the semi-infinite boundary condition, stationary-state trajectories are naturally parametrized by the boundary slope
$s=\partial_x \psi(0)\in\mathbb{R}$ [Eq.~\eqref{eq:bc_slope}]. 
To quantify which long-$x$ behavior is ``typical'' at a given $\gamma$, we consider an ensemble of trajectories with initial slopes sampled from a probability measure $\mu$ on $\mathbb{R}$ with density $\rho(s)$, and define the basin-fraction order parameter
\begin{equation}
\begin{aligned}
p_{\mathrm{skin}}(\gamma;\mu)&\coloneqq 
\Pr_{s\sim\mu}\!\Bigl[\text{trajectory is attracted to the origin}\Bigr]\\
&=
\int_{\mathbb{R}}\rho(s)\,\mathbf{1}_{\mathrm{skin}}(s;\gamma)\,ds,
\label{eq:pskin_def}
\end{aligned}
\end{equation}
where $\mathbf{1}_{\mathrm{skin}}(s;\gamma)\in\{0,1\}$ indicates whether the trajectory launched from $(\psi,v)=(0,s)$ flows to the skin attractor.
Throughout, we focus on admissible measures with continuous densities and unbounded support; such choices avoid artificial cutoffs on $s$ and do not alter the qualitative scenario described below.

To compute $p_{\mathrm{skin}}$, the key ingredient is the separatrix structure in the coexistence regime. 
For $\gamma_c<\gamma<0$, the stable origin and the stable outer limit cycle coexist, and their basins are separated by the unstable inner limit cycle (Sec.~\ref{sec:bif_diag}). 
Projecting this separatrix onto the boundary-slope axis yields a threshold $s_{*}(\gamma)>0$:
initial slopes with $|s|<s_{*}(\gamma)$ lie inside the separatrix and flow to the origin (skin), whereas those with $|s|>s_{*}(\gamma)$ lie outside and are attracted to the outer cycle (extended).
Consequently, in the coexistence window one has the general relation
\begin{equation}
p_{\mathrm{skin}}(\gamma;\mu)
=
\int_{-s_{*}(\gamma)}^{s_{*}(\gamma)} \rho(s)\,ds,
\qquad \gamma_c<\gamma<0,
\label{eq:pskin_cdf}
\end{equation}
while outside this window the attractor is unique and no separatrix exists: $p_{\mathrm{skin}}=1$ for $\gamma<\gamma_c$ (skin-only) and $p_{\mathrm{skin}}=0$ for $\gamma>0$ (extended-only).

Eq.~\eqref{eq:pskin_def}--\eqref{eq:pskin_cdf} directly translate the bifurcation skeleton into a basin-geometry-induced phase-transition scenario.
At $\gamma=\gamma_c$ the SNLC bifurcation creates an outer stable cycle at a \emph{finite} amplitude, together with an inner unstable cycle that immediately forms a separatrix~\cite{kuznetsov1998elements}.
As a result, a finite measure of initial slopes is captured by the extended attractor right at onset, producing a \emph{first-order-like jump} in $p_{\mathrm{skin}}$. 
This jump is robust for any admissible probability density $\rho$ with full support, and its size is 
$\Delta p_{\mathrm{skin}} = \int_{-\infty}^{-s_{*}(\gamma_c)} \rho(s)\,ds +  \int_{s_{*}(\gamma_c)}^{\infty} \rho(s)\,ds$.
Within $\gamma_c<\gamma<0$, the separatrix shrinks toward the origin as $\gamma\uparrow 0$, so the skin basin contracts and $p_{\mathrm{skin}}$ decreases continuously.
Finally, as $\gamma\to 0^{-}$ the unstable inner cycle collapses into the origin at the subcritical Hopf point, the origin loses stability, and the skin outcome becomes unreachable as a long-$x$ attractor, yielding $p_{\mathrm{skin}}\to 0$.

To visualize these basin-geometry quantities, Fig.~\ref{fig:basin_frac} shows both the boundary-slope threshold $s_{*}(\gamma)$ and the order parameter $p_{\mathrm{skin}}(\gamma)$ computed numerically.
In practice, we choose a representative probability measure on slopes by the compactification
$s=s_0\tan\theta$ with $\theta\sim\mathrm{Unif}(-\pi/2,\pi/2)$ and $s_0=\sqrt{2E}$,
equivalently the Cauchy density $\rho_{\rm c}(s)=\frac{1}{\pi}\frac{s_0}{s^2+s_0^2}$ of unbounded support.
Other admissible choices of $\mu$ lead to the same qualitative phase diagram, particularly a basin-geometry-induced first-order-like jump at SNLC, while mainly changing the value of $p_{\mathrm{skin}}(\gamma)$ inside the bistable region.

\begin{figure}
  \includegraphics[width=0.95\linewidth]{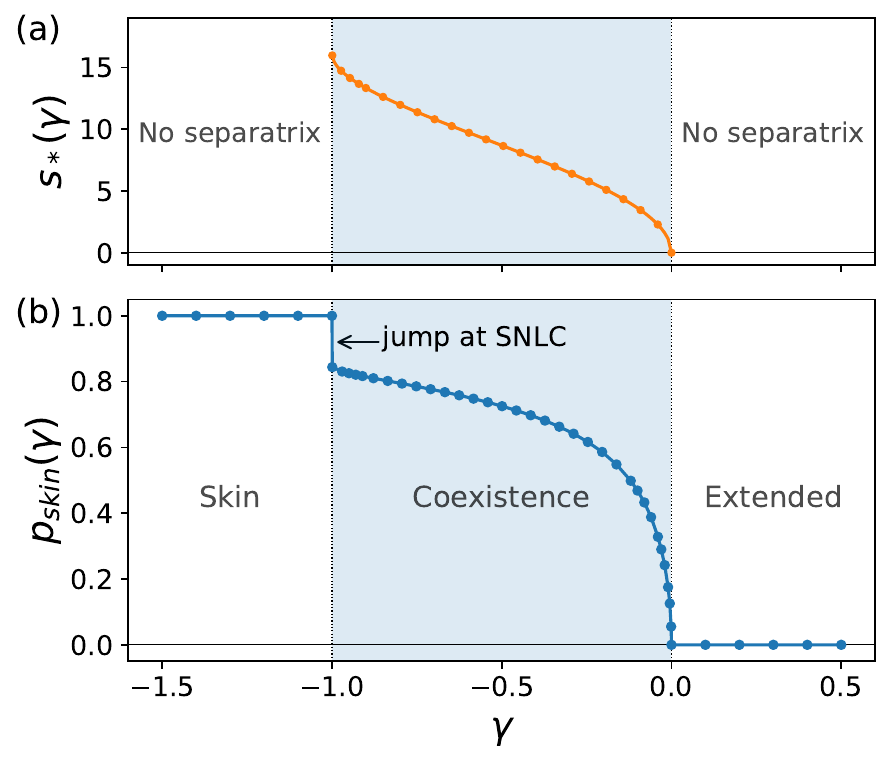}
  \caption{Basin-geometry characterization of the phase diagram.
  (a) Threshold $s_{*}(\gamma)$ separating boundary slope ($s=\partial_x\psi(0)$) that flow to the origin from those attracted to the outer limit cycle. In the coexistence window $\gamma_c<\gamma<0$ (shaded), the separatrix is the unstable inner limit cycle, whose projection onto the boundary-slope axis yields $s_{*}(\gamma)$. Outside this window the attractor is unique and no separatrix exists.
  (b) Basin-fraction order parameter $p_{\mathrm{skin}}(\gamma)$. A first-order-like jump occurs at SNLC.
  Parameters are $a=1/2$, $b=1/32$ and $E=8$. 
  \label{fig:basin_frac} 
  }
\end{figure}

\section{Physical phenomena in the coexistence window}
\label{sec:coexist_phys}

The coexistence window $\gamma_c<\gamma<0$ supports two distinct stable long-$x$ outcomes under the same system parameters: the fixed point at the origin (skin mode) and the stable limit cycle (extended state).
These attractors are separated by the unstable inner limit cycle, which acts as a separatrix in phase space (Sec.~\ref{sec:bif_diag}).
Beyond providing a static classification of stationary states, this basin geometry has direct physical consequences.
In this section, we highlight two generic manifestations: (i) long-lived oscillatory transients in the spatial evolution and sharp initial-condition sensitivity near the separatrix, and (ii) hysteresis (memory) expected under adiabatic parameter ramps due to bistability.

\subsection{Separatrix-induced long-lived spatial transients}
\label{sec:transient_separatrix}

Besides the two asymptotic stationary states (skin and extended), the coexistence window $\gamma_c<\gamma<0$ exhibits a pronounced \emph{long-lived transient} in the spatial evolution when the trajectory is initialized near the basin separatrix.
In this regime, the unstable inner limit cycle forms the separatrix between the origin and the stable outer cycle (Sec.~\ref{sec:bif_diag}), and its projection onto the boundary-slope axis defines the threshold $s_*(\gamma)$ (Sec.~\ref{sec:basin_frac}).
Trajectories launched with $s$ close to $s_*(\gamma)$ can shadow the unstable inner cycle for a long distance in $x$ before eventually approaching one of the stable attractors.

This separatrix-induced transient is directly visible in Fig.~\ref{fig:phase}(e,f).
In both cases, the wave function first develops an extended oscillatory segment whose envelope stays close to the inner-cycle amplitude over a large range of $x$, and then settles into its final long-$x$ behavior---decay to zero or growth to a stable extended state with larger amplitude.
Importantly, this transient-skin state [Fig.~\ref{fig:phase}(e)] is qualitatively distinct from a conventional skin mode, which typically decays exponentially away from the boundary.
Here, the transient is a consequence of the bistability--separatrix geometry in the nonlinear phase space. This mechanism has no counterpart in the linear case.

\subsection{Hysteresis under adiabatic parameter ramps}
\label{sec:hysteresis}

The bistable structure also implies a robust hysteresis (memory) effect when the nonreciprocity parameter $\gamma$ is varied adiabatically, such that the system can relax to the currently stable attractor at each parameter value.

Consider first an adiabatic \emph{decrease} of $\gamma$ starting from the extended regime $\gamma>0$ where the unique attractor is the stable limit cycle.
As $\gamma$ is lowered into the coexistence window, the system can remain on the extended branch because this attractor persists for $\gamma_c<\gamma<0$.
Upon reaching $\gamma=\gamma_c$, the stable and unstable cycles annihilate in a SNLC, and the extended attractor disappears; the system must then switch to the only remaining attractor, the origin, yielding a skin state.

Conversely, under an adiabatic \emph{increase} of $\gamma$ starting from the skin regime $\gamma<\gamma_c$, the trajectory remains attracted to the origin throughout the coexistence window because the origin stays stable for all $\gamma<0$.
Only at $\gamma=0$ does the origin lose stability through the subcritical Hopf bifurcation, beyond which the system switches to the extended state.
Therefore, the switching points differ depending on the sweep direction: $\gamma_c$ on the downward sweep (extended $\to$ skin) and $0$ on the upward sweep (skin $\to$ extended), forming a hysteresis loop.
This provides a natural ``memory'' effect: within $\gamma_c<\gamma<0$, the observed stationary behavior depends on the preparation history, not solely on the instantaneous parameter values.

A full dynamical characterization of this hysteresis---including finite-rate effects (bifurcation delay), noise-induced switching near the separatrix, and concrete implementations in finite systems with physical boundary conditions---requires solving an explicit time-dependent model and lies beyond the scope of the present work.
Here we emphasize that the essential ingredients of hysteresis (bistability and distinct stability-loss thresholds) are already established at the stationary level by the global bifurcation diagram and the basin geometry described in Secs.~\ref{sec:bif_diag}--\ref{sec:basin_frac}.

\section{Conclusions}\label{sec:concl}
In this work we studied a nonlinear non-Hermitian model with a saturating amplitude-dependent nonreciprocity.
Casting the stationary equation into a planar phase-space flow establishes an attractor--localization correspondence.
We uncovered a global bifurcation scenario governed by two thresholds.
The origin loses stability via a \emph{subcritical} Hopf bifurcation at $\gamma=0$, while a SNLC bifurcation occurs at $\gamma=\gamma_c<0$, creating an outer stable and an inner unstable cycle.
Together they generate a finite coexistence window $\gamma_c<\gamma<0$, where skin and extended stationary states are simultaneously stable at a fixed energy $E$ and are separated by a basin boundary organized by the unstable cycle---a coexistence mechanism distinct from linear spectral (mobility-edge) scenarios.

To provide an analytical backbone, we derived an averaged one-dimensional amplitude equation that predicts the inner/outer limit-cycle branches and the SNLC threshold.
Building on the basin geometry, we introduced a basin-fraction order parameter that yields a phase diagram and exhibits a first-order-like jump at SNLC.
We further discussed physical consequences of bistability in the coexistence window, including separatrix-induced long-lived transients and hysteresis under adiabatic parameter sweeps.

These results enrich the nonlinear non-Hermitian skin effect by highlighting the role of global attractor--basin geometry in understanding stationary states in nonlinear non-Hermitian systems, and motivate future studies of time-dependent dynamics, noise-induced switching, and finite-size implementations under physical boundary conditions.

\begin{acknowledgments}
This work is supported by the National Natural Science Foundation of China under Grants No. 11974205, No. 61727801, and No. 62131002, and the Key Research and Development Program of Guangdong Province (Grant No. 2018B030325002).
\end{acknowledgments}

\appendix
\section{Numerical methods}\label{app:numerics}

This appendix summarizes the numerical methods used in
Figs.~\ref{fig:bif_diag}--\ref{fig:basin_frac}.
All computations are based on integrating the phase-space flow Eq.~\eqref{eq:phase_flow}
with the semi-infinite boundary condition $\psi(0)=0$ and $v(0)=s$.

\subsection{Shooting integration}\label{app:num_shooting}

For given $(\gamma,a,b,E)$ and initial slope $s$, we solve Eq.~\eqref{eq:phase_flow} on $x\in[0,L]$
using \texttt{solve\_ivp} from \texttt{SciPy}~\cite{2020SciPy-NMeth} with the \texttt{DOP853} integrator.
The integration length $L$ is chosen sufficiently large so that the long-$x$ behavior is insensitive to further increasing $L$.

When a stable limit cycle is present, we measure its amplitude $A$ as the turning-point value
of $|\psi|$ along the asymptotic periodic orbit, i.e., at points where $v=\partial_x\psi=0$.

\subsection{Continuation via Poincar\'e return map}\label{app:num_continuation}

We formulate limit cycles as fixed points of a one-dimensional Poincar\'e return map~\cite{kuznetsov1998elements}.
Define the section
\begin{equation}
\Sigma \coloneqq \{(\psi,v):\ \psi=0,\ v>0\}.
\end{equation}
Starting from $(\psi,v)=(0,s)\in\Sigma$, we integrate Eq.~\eqref{eq:phase_flow} until the
\emph{next} return to $\Sigma$, and denote the return velocity by $v_1$.
This defines the scalar return map
\begin{equation}
P_{\gamma}(s) \coloneqq v_1 .
\end{equation}
A limit cycle corresponds to a fixed point
\begin{equation}
P_{\gamma}(s_{\rm LC}) = s_{\rm LC},
\end{equation}
equivalently a root of the scalar equation
\begin{equation}
F_{\rm LC}(s,\gamma) \coloneqq P_{\gamma}(s)-s = 0 .
\label{eq:app_FLC}
\end{equation}

To trace the limit-cycle branch in $\gamma$, we perform pseudo-arclength continuation~\cite{kuznetsov1998elements} on the
solution set of Eq.~\eqref{eq:app_FLC}. Given a known solution $(s_k,\gamma_k)$ and a local tangent
direction $\mathbf{t}_k=(t_s,t_\gamma)$, we use a predictor
\begin{equation}
(s^{\rm p},\gamma^{\rm p})
=(s_k,\gamma_k)+\Delta \ell\,\mathbf{t}_k,
\end{equation}
followed by a corrector that solves
\begin{equation}
\begin{cases}
F_{\rm LC}(s,\gamma)=0,\\
(t_s,t_\gamma)\!\cdot\!\bigl((s,\gamma)-(s^{\rm p},\gamma^{\rm p})\bigr)=0,
\end{cases}
\end{equation}
by Newton iteration. Once $s_{\rm LC}$ is obtained, we integrate from $(0,s_{\rm LC})$ and
extract the amplitude $A$ from turning points ($v=0$) as in
Appendix~\ref{app:num_shooting}.

\subsection{Basin fraction and separatrix threshold}\label{app:num_basin}

In the coexistence window $\gamma_c<\gamma<0$, the basin boundary is organized by the unstable
inner limit cycle, whose projection to the boundary-slope axis yields a threshold $s_*(\gamma)>0$
(Sec.~\ref{sec:basin_frac}). Numerically, we locate $s_*(\gamma)$ by bisection in $s>0$ using
a skin/extended classifier applied to the shot trajectory.

For classification, we compute an inverse participation ratio on $[0,L]$,
\begin{equation}
{\rm IPR}(L)
\coloneqq
\frac{\int_0^L |\psi(x)|^4\,dx}
{\left(\int_0^L |\psi(x)|^2\,dx\right)^2},
\end{equation}
and the associated fractal dimension~\cite{thouless1974electrons, evers2008anderson}
\begin{equation}
D_2(L) \coloneqq -\frac{\ln({\rm IPR}(L))}{\ln(L)}.
\end{equation}
Extended states have ${\rm IPR}(L)\sim 1/L$ and thus $D_2\to 1$, while skin-localized states have
${\rm IPR}(L)=\mathcal{O}(1)$ and thus $D_2\to 0$ as $L\to\infty$; we use this criterion (with
convergence checked against increasing $L$) to decide the basin outcome during bisection.

After obtaining $s_*(\gamma)$, we compute the basin fraction from
Eq.~\eqref{eq:pskin_cdf}:
\begin{equation}
p_{\mathrm{skin}}(\gamma;\mu)
=
\int_{-s_*(\gamma)}^{s_*(\gamma)} \rho(s)\,ds.
\end{equation}
For the Cauchy density used in Sec.~\ref{sec:basin_frac},
\begin{equation}
\rho(s)=\frac{1}{\pi}\frac{s_0}{s^2+s_0^2},
\end{equation}
the integral has the closed form
\begin{equation}
p_{\mathrm{skin}}(\gamma;\mu)
=
\frac{2}{\pi}\arctan\!\left(\frac{s_*(\gamma)}{s_0}\right).
\end{equation}

\section{A detailed proof of the bifurcation structure}
\label{app:proofs}

In this appendix, we provide a detailed proof of the bifurcation scenario summarized in
Fig.~\ref{fig:bif_diag}. Throughout we assume
\begin{equation}
a>0,\qquad b>0,\qquad E>0,
\label{eq:app_assumptions}
\end{equation}
and consider the one-parameter family in $\gamma\in\mathbb{R}$.

\subsection{Planar Li\'enard form}
\label{app:proofs_lienard}

We start with the stationary equation Eq.~\eqref{eq:stationary_ode} in real space, which can be written as the scalar second-order ODE
\begin{equation}
\partial_x^2\psi
-2\Bigl(\gamma+a\psi^2-b\psi^4\Bigr)\,\partial_x\psi
+2E\,\psi
=0,
\label{eq:app_scalar}
\end{equation}
which is of polynomial Li\'enard type~\cite{lienard1928, guckenheimer2013nonlinear},
\begin{equation}\label{eq:app_lienard}
\partial_x^2\psi+f(\psi)\partial_x\psi+g(\psi)=0    
\end{equation}
with
\begin{equation}
f(\psi)=-2\Bigl(\gamma+a\psi^2-b\psi^4\Bigr),
\qquad
g(\psi)=2E\,\psi.
\label{eq:app_fg}
\end{equation}

\subsection{Lyapunov function and uniform boundedness}
\label{app:proofs_lyapunov}

Taking the phase-space parameterization $v=\partial_x\psi$ and recalling the planar system Eq.~\eqref{eq:phase_flow},
\begin{equation}
\begin{aligned}
\partial_x\psi &= v,\\
\partial_x v &= 2\Bigl(\gamma+a\psi^2-b\psi^4\Bigr)\,v-2E\,\psi.
\end{aligned}
\label{eq:app_planar}
\end{equation}

\begin{lemma}[Lyapunov function~\cite{kuznetsov1998elements}]\label{lemma:Lyapunov}
Define
\begin{equation}
V(\psi,v) \coloneqq \frac{1}{2}v^2+E\psi^2.
\label{eq:app_V}
\end{equation}
This form is analogous to the energy of a harmonic oscillator. Along any solution of Eq.~\eqref{eq:app_planar},
\begin{equation}
\partial_x V
=
2\Bigl(\gamma+a\psi^2-b\psi^4\Bigr)\,v^2.
\label{eq:app_Vdot}
\end{equation}
\end{lemma}

\begin{proof}
Using $\partial_x V=v\,\partial_x v+2E\psi\,\partial_x\psi$ and substituting Eq.~\eqref{eq:app_planar} gives Eq.~\eqref{eq:app_Vdot}.
\end{proof}

\begin{lemma}[Forward boundedness]\label{lemma:boundedness}
For any fixed $\gamma\in\mathbb{R}$, all forward trajectories of Eq.~\eqref{eq:app_planar}
remain bounded in the $(\psi,v)$ plane.
\end{lemma}

\begin{proof}
Choose $R>0$ sufficiently large so that
\begin{equation}
\gamma+a\psi^2-b\psi^4\le 0
\qquad \text{whenever }|\psi|\ge R,
\label{eq:app_R_choice}
\end{equation}
which is possible since the $-b\psi^4$ term dominates for large $|\psi|$.
Then by Eq.~\eqref{eq:app_Vdot},
\begin{equation}
\partial_x V \le 0\qquad\text{whenever }|\psi|\ge R,
\label{eq:app_V_noninc}
\end{equation}
so once a trajectory enters $\{|\psi|\ge R\}$ the Lyapunov function $V$ cannot increase there,
preventing the trajectory from escaping to infinity. Together with continuity of the flow, this implies uniform forward
boundedness.
\end{proof}

\subsection{Strongly negative $\gamma$: global stability of the origin and no periodic orbits}
\label{app:proofs_strong_neg}

\begin{proposition}[Only stable origin for sufficiently negative $\gamma$]
\label{prop:stable_origin}
If
\begin{equation}
\gamma<-\frac{a^2}{4b},
\label{eq:app_gamma_strong}
\end{equation}
then the origin $(\psi,v)=(0,0)$ is globally asymptotically stable and
Eq.~\eqref{eq:app_planar} admits no periodic orbits.
\end{proposition}

\begin{proof}
For $\psi^2\ge 0$, the polynomial $a\psi^2-b\psi^4$ attains its maximum $a^2/(4b)$ at $\psi^2=a/(2b)$.
Hence Eq.~\eqref{eq:app_gamma_strong} implies
\begin{equation}
\gamma+a\psi^2-b\psi^4 \le \gamma+\frac{a^2}{4b}<0
\qquad\forall\,\psi\in\mathbb{R}.
\label{eq:app_F_neg_everywhere}
\end{equation}
Then Eq.~\eqref{eq:app_Vdot} gives $\partial_x V<0$ whenever $v\neq 0$, so the Lyapunov function $V$ strictly decreases along nontrivial trajectories.
By LaSalle's invariance principle~\cite{lasalle1976stability}, every trajectory converges to the largest invariant set
contained in $\{v=0\}$. On $\{v=0\}$ one has $\partial_x v=-2E\psi$, so the only invariant point
is $(0,0)$. Thus the origin is global asymptotically stable.
A periodic orbit would force $V$ to be periodic, contradicting strict decrease of $V$ along
nontrivial trajectories.
\end{proof}

\subsection{Local bifurcation at $\gamma=0$: subcritical Hopf and a small unstable cycle}
\label{app:proofs_hopf}

\begin{proposition}[Hopf bifurcation at $\gamma=0$]
\label{prop:hopf}
The equilibrium $(0,0)$ undergoes a Hopf bifurcation at $\gamma=0$ with linear frequency
$\omega=\sqrt{2E}$.
\end{proposition}

\begin{proof}
The Jacobian of Eq.~\eqref{eq:app_planar} at $(0,0)$ equals
\begin{equation}
J(\gamma)=
\begin{pmatrix}
0 & 1\\
-2E & 2\gamma
\end{pmatrix},
\label{eq:app_J}
\end{equation}
so $\mathrm{tr}\,J(\gamma)=2\gamma$ and $\det J(\gamma)=2E>0$.
At $\gamma=0$ the eigenvalues are $\pm i\sqrt{2E}$, and the transversality condition
$\partial_\gamma(\mathrm{tr}\,J)|_{\gamma=0}=2\neq 0$ holds, hence a Hopf bifurcation occurs. 
\end{proof}

\begin{proposition}[Subcritical Hopf for $a>0$]
\label{prop:subcritical}
Assume $a>0$. Then the Hopf bifurcation at $\gamma=0$ is \emph{subcritical}:
there exists $\varepsilon>0$ such that for every $\gamma\in(-\varepsilon,0)$ the system admits
a small-amplitude \emph{unstable} limit cycle $\Gamma_{\mathrm{in}}(\gamma)$ surrounding the origin.
Moreover, its amplitude satisfies
\begin{equation}
A_{\mathrm{in}}^2(\gamma)= -\frac{4\gamma}{a}+\mathcal{O}(\gamma^2),
\qquad \gamma\to 0^-.
\label{eq:app_rin_scaling}
\end{equation}
\end{proposition}

\begin{proof}
Fix $\gamma<0$ with $|\gamma|$ small.
In the Hopf neighborhood one has $A_{\mathrm{in}}=\mathcal{O}(\sqrt{|\gamma|})$~\cite{kuznetsov1998elements},
hence the slowly varying amplitude condition Eq.~\eqref{eq:fast_slow_cond} is satisfied for sufficiently small $|\gamma|$ (since $\omega=\sqrt{2E}>0$ is fixed).
Therefore the averaged amplitude equation Eq.~\eqref{eq:amplitude_avg} applies (Appendix~\ref{app:avg_eqn}):
\begin{equation}
\partial_x r = r\left(\gamma+\frac{a}{4}r^2-\frac{b}{8}r^4\right),
\label{eq:app_avg_ref}
\end{equation}
which has a nonzero equilibrium $r=A_{\mathrm{in}}(\gamma)>0$
satisfying
\begin{equation}
\gamma+\frac{a}{4}A_{\mathrm{in}}^2-\frac{b}{8}A_{\mathrm{in}}^4=0,
\label{eq:app_rin_root}
\end{equation}
which yields the scaling Eq.~\eqref{eq:app_rin_scaling}.
Moreover, stability follows from the one-dimensional criterion:
for the radial drift $h(r)=r(\gamma+\frac{a}{4}r^2-\frac{b}{8}r^4)$,
\begin{equation}
h'(A_{\mathrm{in}})=
2\gamma+\frac{a}{2}A_{\mathrm{in}}^2-\frac{3b}{4}A_{\mathrm{in}}^4
=
-2\gamma+\mathcal{O}(\gamma^2)>0,
\label{eq:app_rin_unstable}
\end{equation}
so the corresponding cycle is repelling (unstable), which is the defining signature of a
subcritical Hopf bifurcation. 
\end{proof}

\subsection{For $\gamma>0$: unique stable limit cycle from Li\'enard's theorem}
\label{app:proofs_gamma_pos}

\begin{proposition}[Existence and uniqueness of a stable limit cycle for $\gamma>0$]
\label{prop:stable_LC}
For each $\gamma>0$, the planar system Eq.~\eqref{eq:app_planar} admits a unique limit cycle, and
this limit cycle is attracting.
\end{proposition}

\begin{proof}
We apply the classical Li\'enard theorem~\cite{lienard1928, guckenheimer2013nonlinear} to Eq.~\eqref{eq:app_scalar} with $f,g$ given in Eq.~\eqref{eq:app_fg}.
Note that $g(\psi)=2E\psi$ is odd and satisfies $\psi g(\psi)=2E\psi^2>0$ for $\psi\neq 0$.
Define
\begin{equation}
F_{\mathrm{L}}(\psi)
\coloneqq
\int_0^\psi f(s)\,ds
=
-2\left(
\gamma\psi+\frac{a}{3}\psi^3-\frac{b}{5}\psi^5
\right),
\label{eq:app_FL}
\end{equation}
which is an odd quintic polynomial.
For $\gamma>0$, one has $F_{\mathrm{L}}(\psi)<0$ for sufficiently small $\psi>0$,
while $F_{\mathrm{L}}(\psi)\to +\infty$ as $\psi\to +\infty$ (the $+\frac{2b}{5}\psi^5$ term dominates).
Moreover,
\begin{equation}
F_{\mathrm{L}}'(\psi)=f(\psi)=-2(\gamma+a\psi^2-b\psi^4)
\end{equation}
is an even quartic polynomial that becomes strictly positive for all sufficiently large $|\psi|$
(since $2b\psi^4$ dominates).
Therefore there exists a unique $a_*>0$ such that
$F_{\mathrm{L}}(\psi)<0$ for $\psi\in(0,a_*)$ and
$F_{\mathrm{L}}(\psi)>0$ with $F_{\mathrm{L}}'(\psi)>0$ for all $\psi>a_*$.
These are precisely the hypotheses of the Li\'enard theorem, which implies existence and uniqueness
of a limit cycle; attractivity is part of the theorem (equivalently, follows from the Poincar\'e map
multiplier being less than $1$ in magnitude).
\end{proof}

\subsection{Global assembly and the SNLC at $\gamma=\gamma_c$}
\label{app:proofs_global}

We now combine the previous results to obtain the global bifurcation structure.

\begin{theorem}[Global bifurcation scenario]
\label{theorem:bif_diag}
Assume Eq.~\eqref{eq:app_assumptions}. Consider the family Eq.~\eqref{eq:app_planar} as $\gamma$ varies.
\begin{enumerate}
\item \textbf{(Skin-only for strongly negative $\gamma$).}
If $\gamma<-\frac{a^2}{4b}$, the origin is globally asymptotically stable and there are no limit
cycles (Proposition~\ref{prop:stable_origin}).

\item \textbf{(Subcritical Hopf at $\gamma=0$).}
At $\gamma=0$ the origin undergoes a Hopf bifurcation (Proposition~\ref{prop:hopf}), and for $a>0$ it is
subcritical: for $\gamma<0$ sufficiently close to $0$ there exists an \emph{unstable} small cycle
$\Gamma_{\mathrm{in}}(\gamma)$ (Proposition~\ref{prop:subcritical}).

\item \textbf{(Unique outer cycle for $\gamma>0$).}
For every $\gamma>0$ there exists a unique \emph{attracting} limit cycle
$\Gamma_{\mathrm{out}}(\gamma)$ (Proposition~\ref{prop:stable_LC}).

\item \textbf{(Existence of $\gamma_c<0$ and SNLC).}
Define
\begin{equation}
\begin{aligned}
\Gamma \coloneqq & \{\gamma\in\mathbb{R}:
\\
&\quad \text{Eq.~}\eqref{eq:app_planar}\ \text{has at least one limit cycle}\},
\\
\gamma_c \coloneqq& \inf \Gamma.
\end{aligned}
\label{eq:app_gamma_c_def}
\end{equation}
Then $\gamma_c\in\bigl(-\frac{a^2}{4b},\,0\bigr)$.
Moreover, generically (under standard nondegeneracy conditions on the Poincar\'e return map),
the family of limit cycles terminates at $\gamma=\gamma_c$ through a
\emph{SNLC bifurcation}:
at $\gamma=\gamma_c$ there exists a nonhyperbolic (semistable) cycle with Floquet multiplier $+1$,
and for $\gamma>\gamma_c$ close to $\gamma_c$ there are two cycles of opposite stability
(one stable and one unstable), while for $\gamma<\gamma_c$ no cycle exists.

\item \textbf{(At most two cycles).} For any $\gamma$, the number of limit cycles is at most two.

\item \textbf{(Two-cycle region and persistence to $\gamma>0$).} In the parameter window $\gamma_c<\gamma<0$, the origin is locally asymptotically stable (since $\operatorname{tr}J=2\gamma<0$), and one has the unstable inner cycle $\Gamma_{\mathrm{in}}(\gamma)$ plus the stable outer cycle $\Gamma_{\mathrm{out}}(\gamma)$. For $\gamma>0$, the origin is unstable and the unique stable limit cycle is precisely the continuation of $\Gamma_{\mathrm{out}}(\gamma)$.
\end{enumerate}
\end{theorem}

\begin{proof}

Items (1)--(3) follow directly from Propositions~\ref{prop:stable_origin}--\ref{prop:stable_LC}.

To see item (4), note first that (1) implies $\Gamma$ is bounded below by $-a^2/(4b)$,
while (3) implies $\Gamma$ is nonempty and contains all $\gamma>0$.
Hence $\gamma_c$ defined in Eq.~\eqref{eq:app_gamma_c_def} satisfies
$\gamma_c\in(-\frac{a^2}{4b},0]$; and (2) guarantees the existence of a limit cycle for some
$\gamma<0$ sufficiently close to $0$, so in fact $\gamma_c<0$.
For the fold mechanism, fix a transversal section $\Sigma$ and denote by $P_\gamma:\Sigma\to\Sigma$
the Poincar\'e return map. A limit cycle corresponds to a fixed point $\xi=P_\gamma(\xi)$.
A \emph{hyperbolic} cycle corresponds to $P_\gamma'(\xi)\neq 1$ and hence persists smoothly
under parameter variation by the implicit function theorem~\cite{kuznetsov1998elements}.
Therefore, if a branch of limit cycles ceases to exist at a finite parameter value,
the terminal cycle must be nonhyperbolic with multiplier $+1$, i.e. $P_{\gamma_c}'(\xi_c)=1$.
In a one-parameter family, the generic codimension-one way for limit cycles to be
created/annihilated at $P_{\gamma_c}'(\xi_c)=1$ is a SNLC bifurcation:
\begin{equation}
\begin{gathered}
P_\gamma(\xi)-\xi
=
\alpha(\gamma-\gamma_c)+\beta(\xi-\xi_c)^2+\text{h.o.t.},
\\
\alpha\beta\neq 0,
\label{eq:app_SN_map}    
\end{gathered}
\end{equation}
which yields two limit cycles for $\gamma>\gamma_c$ (one stable, one unstable) and none for
$\gamma<\gamma_c$. Translating back to the flow gives the SNLC structure stated by (4).

Item (5) follows directly from a classical result~\cite{max_num_LC}, which implies that the number of limit cycles is at most two for odd quintic Li\'enard systems of the form Eq.~\eqref{eq:app_scalar}.

Combining items (2)-(5) gives item (6).

\end{proof}

Theorem~\ref{theorem:bif_diag} yields the exact bifurcation diagram:
\[
\begin{aligned}
&\gamma<\gamma_c:\ \text{only stable origin};  \\
&\gamma=\gamma_c:\  \text{SNLC bifurcation}; \\
\gamma_c<&\gamma<0:\ \text{stable origin + unstable inner cycle} \\
&\qquad\qquad \text{+ stable outer cycle}; \\
&\gamma = 0: \  \text{subcritical Hopf bifurcation}; \\
&\gamma>0:\ \text{unstable origin + unique stable outer cycle}.
\end{aligned}
\]
This is precisely the bifurcation diagram in Fig.~\ref{fig:bif_diag}.

\section{Detailed derivation of the averaged amplitude equation}
\label{app:avg_eqn}

In this appendix, we provide a detailed derivation of the averaged amplitude equation
used in Sec.~\ref{sec:ave_eqn}.
We start from the planar phase-space flow dynamics Eq.~\eqref{eq:phase_flow},
\begin{equation}
\begin{aligned}
\partial_x \psi &= v,\\
\partial_x v &= 2F(\psi^2)\,v - 2E\,\psi,
\end{aligned}
\label{eq:app_flow_psi_v}
\end{equation}
where $F(\psi^2)=\gamma + a\psi^2 - b\psi^4 .$ and we focus on $E>0$ with real $\psi(x)$.

In the Hermitian limit ($F(\psi^2)=0$ in Eq.~\eqref{eq:app_flow_psi_v}), our model reduces to a harmonic oscillator with frequency
\begin{equation}\label{eq:app_natural_freq}
\omega=\sqrt{2E}.
\end{equation}
It is therefore natural to rescale $v=\partial_x \psi$ as
\begin{equation}
w \coloneqq \frac{v}{\omega}.
\label{eq:app_omega_w}
\end{equation}
Then Eq.~\eqref{eq:app_flow_psi_v} becomes
\begin{equation}
\begin{aligned}
\partial_x \psi &= \omega\, w,\\
\partial_x w &= -\omega\,\psi + 2F(\psi^2)\,w.
\end{aligned}
\label{eq:app_flow_scaled}
\end{equation}
We then introduce polar coordinates in the $(\psi,w)$ plane,
\begin{equation}
\psi = r\cos\theta,\qquad
w = -r\sin\theta,
\qquad r\ge 0,
\label{eq:app_polar}
\end{equation}
so that $r$ measures the oscillation amplitude and $\theta$ is the oscillation phase.

\subsection*{Exact radial dynamics}

A convenient way to obtain the radial equation is to use the identity
\begin{equation}
r^2 = \psi^2 + w^2.
\label{eq:app_r2_def}
\end{equation}
Differentiating Eq.~\eqref{eq:app_r2_def} and using Eq.~\eqref{eq:app_flow_scaled} gives
\begin{equation}
\begin{aligned}
\partial_x(r^2)
&= 2\psi\,\partial_x\psi + 2w\,\partial_x w\\
&= 2\psi(\omega w)
 + 2w\!\left(-\omega\psi + 2F(\psi^2)w\right)\\
&= 4F(\psi^2)\,w^2.
\end{aligned}
\label{eq:app_r2_exact_1}
\end{equation}
Substituting Eq.~\eqref{eq:app_polar} (so that $\psi^2=r^2\cos^2\theta$ and $w^2=r^2\sin^2\theta$),
we obtain the exact identity
\begin{equation}
\partial_x(r^2)
=
4\Bigl(\gamma + a r^2\cos^2\theta - b r^4\cos^4\theta\Bigr)\,
r^2\sin^2\theta.
\label{eq:app_r2_exact_2}
\end{equation}
For $r>0$, dividing by $2r$ yields the exact radial equation
\begin{equation}
\partial_x r
=
2r\Bigl(\gamma + a r^2\cos^2\theta - b r^4\cos^4\theta\Bigr)\sin^2\theta.
\label{eq:app_r_exact}
\end{equation}
This is Eq.~\eqref{eq:r_exact} in the main text. Eq.~\eqref{eq:app_r_exact} shows that the amplitude change is modulated by trigonometric factors, i.e., it oscillates within each cycle.

\subsection*{Exact phase dynamics and \\ the slowly varying amplitude condition}

To obtain the phase equation, we use the standard polar-angle identity
\begin{equation}
\partial_x\theta
=
-\frac{\psi\,\partial_x w - w\,\partial_x\psi}{r^2}.
\label{eq:app_theta_identity}
\end{equation}
Using Eq.~\eqref{eq:app_flow_scaled}, we compute
\begin{equation}
\begin{aligned}
\psi\,\partial_x w - w\,\partial_x\psi
&=
\psi\!\left(-\omega\psi + 2F(\psi^2)w\right) - w(\omega w)\\
&=
-\omega(\psi^2+w^2) + 2F(\psi^2)\psi w\\
&=
-\omega r^2 + 2F(\psi^2)\psi w.
\end{aligned}
\label{eq:app_theta_numer}
\end{equation}
Substituting Eq.~\eqref{eq:app_theta_numer} into Eq.~\eqref{eq:app_theta_identity}
and using $\psi w/r^2 = -\cos\theta\,\sin\theta$ from Eq.~\eqref{eq:app_polar}, we get the exact phase equation
\begin{equation}
\partial_x\theta
=
\omega
+
2F\!\left(r^2\cos^2\theta\right)\sin\theta\cos\theta.
\label{eq:app_theta_exact}
\end{equation}
Eq.~\eqref{eq:app_theta_exact} shows that the leading term of the exact angular velocity is the natural Hermitian frequency $\omega$ from Eq.~\eqref{eq:app_natural_freq}.
Since
\begin{equation}
F\!\left(r^2\cos^2\theta\right)
=
\gamma + a r^2\cos^2\theta - b r^4\cos^4\theta,
\label{eq:app_F_polar}
\end{equation}
the correction term in Eq.~\eqref{eq:app_theta_exact} is bounded by
$\mathcal{O}(|\gamma|+a r^2+br^4)$.
Therefore, in the regime
\begin{equation}
|\gamma| + ar^2 + br^4 \ll \omega=\sqrt{2E},
\label{eq:app_fast_slow}
\end{equation}
the phase advances almost uniformly, $\partial_x\theta \approx \omega$,
while the radius changes slowly over one oscillation period, $|r'|/r\ll |\theta'|$.
This is the ``slowly varying amplitude'' condition in Eq.~\eqref{eq:fast_slow_cond} used in the main text~\cite{2007Averaging, Kevorkian1996Multiple}.

\subsection*{First-order averaging and the averaged amplitude equation}

Under Eq.~\eqref{eq:app_fast_slow}, one may apply the standard first-order averaging theorem
to the radial equation Eq.~\eqref{eq:app_r_exact} by averaging over $\theta\in[0,2\pi)$.
Using the elementary averages
\begin{equation}
\begin{gathered}
\Bigl\langle \sin^2\theta \Bigr\rangle = \frac{1}{2},\\
\Bigl\langle \cos^2\theta\,\sin^2\theta \Bigr\rangle = \frac{1}{8},\\
\Bigl\langle \cos^4\theta\,\sin^2\theta \Bigr\rangle = \frac{1}{16},
\end{gathered}
\label{eq:app_trig_avgs}
\end{equation}
the averaged drift of $r$ becomes
\begin{equation}
\partial_x r
\approx
r\left(\gamma+\frac{a}{4}r^2-\frac{b}{8}r^4\right).
\label{eq:app_avg_drift}
\end{equation}
More precisely, first-order averaging yields
\begin{equation}
\begin{aligned}
\partial_x r
=
&r\left(\gamma+\frac{a}{4}r^2-\frac{b}{8}r^4\right)
\\
&+
\mathcal{O}\!\left(\frac{r\left(|\gamma|+ar^2+br^4\right)^2}{\omega}\right),
\end{aligned}
\label{eq:app_avg_with_remainder}
\end{equation}
where the remainder is higher order in the small parameter
$\left(|\gamma|+ar^2+br^4\right)/\omega$~\cite{2007Averaging, Kevorkian1996Multiple}.
Dropping the remainder term gives Eq.~\eqref{eq:amplitude_avg} in the main text:
\begin{equation}
\partial_x r = h(r),
\qquad
h(r)\equiv r\left(\gamma+\frac{a}{4}r^2-\frac{b}{8}r^4\right).
\label{eq:app_avg_eqn_final}
\end{equation}

Finally, note that at a turning point of $\psi(x)$ on a periodic orbit one has $v=\partial_x\psi=0$,
equivalently $w=0$ and $\sin\theta=0$ in Eq.~\eqref{eq:app_polar}, so $|\psi|=r$.
Thus the equilibria $r=A$ correspond directly to the limit-cycle amplitudes used in
Sec.~\ref{sec:bif_diag}.

\bibliography{apssamp}
\end{document}